\documentclass[11pt]{article}

\usepackage[utf8]{inputenc}
\usepackage{amsthm}
\usepackage{amsmath,amssymb}
\usepackage{cleveref}
\usepackage{fullpage}
\usepackage{marginnote}
\usepackage{xcolor}
\usepackage{enumerate}
\usepackage{cite}
\usepackage{comment}
\newtheorem{theorem}{Theorem}[section]
\newtheorem{lemma}[theorem]{Lemma}
\newtheorem{claim}[theorem]{Claim}

\newtheorem{conjecture}[theorem]{Conjecture}
\newtheorem{definition}[theorem]{Definition}

\newcommand{\R}{\mathbb{R}}

\newcommand{\Z}{\mathbb{Z}}
\newcommand{\D}{\mathcal{D}}
\newcommand{\eps}{\varepsilon}

\newcommand{\SD}{\mathcal{S}_D}
\newcommand{\N}{\mathcal{N}}
\newcommand{\ND}{\mathcal{N}_D}
\newcommand{\F}{\mathbb{F}}

\newcommand{\E}{\mathbb{E}}

\newcommand{\ip}[1]{\langle #1 \rangle}
\newcommand{\LCD}{\text{LCD}}

\title{Singularity of random integer matrices with large entries}
\author{
Sankeerth Rao Karingula\thanks{email: sankeerth1729@gmail.com. Research supported by NSF CCF award 1614023.}
\qquad
Shachar Lovett\thanks{email: slovett@ucsd.edu. Research supported by NSF CCF award 1614023.}\\
Department of Computer Science\\
University of California, San Diego
}

\begin{document}
\maketitle
\begin{abstract}
We study the singularity probability of random integer matrices. Concretely, the probability that a random $n \times n$ matrix, with integer entries chosen uniformly from $\{-m,\ldots,m\}$, is singular. This problem has been well studied in two regimes: large $n$ and constant $m$; or large $m$ and constant $n$. In this paper, we extend previous techniques to handle the regime where both $n,m$ are large. We show that the probability that such a matrix is singular is $m^{-cn}$ for some absolute constant $c>0$. We also provide some connections of our result to coding theory.
\end{abstract}

\section{Introduction}
In this paper we study the probability that a random $n \times n$ matrix with uniform integer entries in $\{-m,\ldots,m\}$ is singular. 
Note that the probability that such a matrix is singular is at least $(2m+1)^{-n}$, which is the probability that its first two rows are the same. We show that this bound is tight, up to polynomial factors.

\begin{theorem}[Singularity of random matrices]
\label{main}
Let $n,m \ge 1$. Let $M$ be an $n \times n$ random integer matrix with entries chosen uniformly in $\{-m,\ldots,m\}$. Then for some absolute constant $c>0$,
$$
\Pr[M \text{ is singular}] \le m^{-cn}.
$$
\end{theorem}

Note that if instead we chose $M$ to be a random $n \times n$ matrix over a finite field $\F_q$, then the probability that $M$ is singular would be about $1/q$, independent of how large $n$ is. This is the main point of difference between random matrices over integers and over finite fields - the singularity probability over integers decreases as the matrix becomes larger, whereas over finite fields it stabilizes.

\subsection{Connections to coding theory - alphabet size for MDS codes}
Our motivation for proving \Cref{main} is a connection to coding theory. More specifically, the question of what alphabet size is needed for Maximum Distance Separable (MDS) codes over the integers.

Coding theory is the study of error correction codes. Codes are widely used in many applications, such as data compression, cryptography, error detection and correction, data transmission and data storage. Algorithms needed to implement codes perform arithmetic operations over an underlying alphabet, and hence their computational complexity is constrained by this alphabet size. Thus, understanding the alphabet size needed to support a given code structure is a central question in coding theory. By far, the most common approach to design codes is to use linear codes over finite fields. Our results in this paper help with investigating the possibility of designing codes over integers. In particular, we study the alphabet size needed to support basic code structures, and focus on the most basic and well-studied family of codes - Maximum Distance Separable (MDS) codes.

An MDS code is a code with the best possible tradeoff between the message length, codeword length and minimal distance. Concretely, an $(n,k,d)$-code is a code with message length $k$, codeword length $n$ and minimal distance $d$. The Singleton bound \cite{singleton1964maximum} gives that $d \le n-k+1$. MDS codes are codes achieving this bound, namely $(n,k,d)$-codes with $d=n-k+1$.
If we consider linear codes, then it is well-known that MDS codes are generated by the row span of \emph{MDS matrices}.

\begin{definition}[MDS matrix]
Let $n \ge k$. A $k \times n$ matrix is called an \emph{MDS matrix} if any $k$ columns in it are linearly independent. Equivalently, if any $k \times k$ minor of it is nonsingular.
\end{definition}

Note that MDS matrices can be defined over finite fields or over the integers. If we define them over a finite field $\F_q$, then
it is well-known that a linear field size is needed to support MDS matrices. Concretely, if we assume $n \ge k+2$, then it is known that
$q \ge \max(k, n-k+1)$ (see for example the introduction of \cite{ball2012sets} for a proof). In particular, this implies that $q \ge n/2$.
Reed-Solomon codes can be constructed over fields of size $q \ge n-1$, which is tight up to a factor of two.
The MDS conjecture of Segre \cite{segre1955curve} speculates that this is indeed the best possible (except for a few special cases), and Ball \cite{ball2012sets} proved this over prime finite fields. In summary, over finite fields a linear field size $q=\Theta(n)$ is both necessary and sufficient.

We show that over the integers, MDS matrices exist over much smaller alphabet sizes.

\begin{theorem}[MDS matrices over integers]
\label{mds_integer}
Let $n \ge k$. There exist $k \times n$ MDS matrices over integers whose entries are in $\{-m,\ldots,m\}$, where $m \le  (cn/k)^c$ for some absolute constant $c>0$.
\end{theorem}

The typical regime in coding theory is that of linear rate and linear distance; namely, where $k=\alpha n$ for some constant $\alpha \in (0,1)$. Note that in this regime \Cref{mds_integer} shows that MDS codes over the integers exist with a \emph{constant} alphabet size, which is in stark contrast with the case over finite fields. It is easy to see that \Cref{mds_integer} is best possible, up to the unspecified constant $c$.

\Cref{mds_integer} follows directly from \Cref{main}.

\begin{proof}[Proof of \Cref{mds_integer}]
Let $M$ be a random $k \times n$ matrix with entries chosen uniformly from $\{-m,\ldots,m\}$. The number of $k \times k$ minors for $M$ is $\binom{n}{k}$, and the probability that each one is singular is at most $m^{-ck}$ by \Cref{main}. Thus
$$
\Pr[M \text{ is not MDS}] \le \binom{n}{k} m^{-ck} \le \left(\frac{en}{k}\right)^k m^{-ck} = \left(\frac{en}{k m^c}\right)^k.
$$
In particular, this probability is at most $2^{-k}$ (say) whenever $m \ge (2en/k)^{1/c}$.
\end{proof}

The following claim shows that the bound in \Cref{mds_integer} is best possible, up to the value of the unspecified constant $c$.

\begin{claim}
\label{mds_lowerbound}
Let $n \ge k \ge 2$.
Let $M$ be a $k \times n$ MDS matrix whose entries are in an alphabet $\Sigma$. Then $|\Sigma| \ge \sqrt{n/k}$.
\end{claim}

\begin{proof}
Let $P_i=(M_{1,i},M_{2,i}) \in \Sigma^2$ denote the first two elements in the $i$-th column of $M$. If $n > |\Sigma|^2 k$, then there must be $k$ distinct columns $i_1,\ldots,i_k \in [n]$ such that $P_{i_1}=\ldots=P_{i_k}$. But then $M$ cannot be an MDS matrix, as the $k \times k$ minor formed by taking these columns has the first two rows being a scalar multiple of each other, and hence cannot be nonsingular.
\end{proof}

Note that the proof of \Cref{mds_integer} is by choosing the matrix $M$ randomly, and showing that with high probability it will be an MDS matrix. This is another aspect in which codes over integers seem to be different from codes over finite fields. Constructing MDS matrices over finite fields seems to require algebraic constructions (such as Reed-Solomon codes), unless the field size is exponential in $n$; whereas over the integers, random matrices work well even for very small entries.

\subsection{Related works on the singularity of random matrices}
Most previous works in random matrix theory focused on random matrices whose entries are sampled independently from distributions with bounded tail behaviour. The most studied case is that of random sign matrices, namely with uniform $\{-1,1\}$ entries. Koml{\'o}s \cite{komlos1967determinant} proved
that the probability that a random $n \times n$ sign matrix is singular is $o(1)$ as $n \to \infty$, which already is a nontrivial result. It took nearly 30 years until
Kahn, Koml{\'o}s and Szemer{\'e}di \cite{kahn1995probability} improved the bound to $c^n$ for some constant $c \in (0,1)$.
A sequence of works \cite{tao2006random,tao2007singularity,bourgain2010singularity} improved the value of the constant $c$, and recently
Tikhomirov \cite{tikhomirov2020singularity} proved that $c=1/2+o(1)$, which is best possible, as the probability
that the first two rows of the matrix are equal is $2^{-n}$. For more general distributions, Rudelson and Vershynin \cite{rudelson2008littlewood,rudelson2010non} proved that if the entries of an $n \times n$ matrix
are sampled from a sub-Gaussian distribution, then the probability it is singular is at most $c^n$ for some $c \in (0,1)$. 

The other regime, of large $m$ and constant $n$, was less explored. The only work we are aware of is by Katznelson \cite{katznelson1993singular} which gave a bound of the form $c_n m^{-n}$
for some constant $c_n$ depending on $n$. While having optimal dependence on $m$ for constant $n$, it has a caveat - it only applies in the regime where $m$ is much larger than $n$
(more precisely, for every fixed $n$, in the limit of large $m$).

A recent work that did address the regime of both large $n$ and large $m$, but for a different entry distribution, is that of Vempala, Wang and Woodruff \cite{vempala2020communication}. Fix $\mu \in (0,1)$, and let $\D_{\mu}$ be a distribution over $\{-1,0,1\}$ with $\D_{\mu}(0)=1-\mu$, $\D_{\mu}(1)=\D_{\mu}(-1)=\mu/2$. Let $M$ be a random $n \times n$ matrix, whose entries are the sum of $m$ independent copies of $\D_{\mu}$.
They show that the probability that $M$ is singular is at most $m^{-cn}$ for some constant $c=c(\mu)>0$. 

With respect to the connection to coding theory, we note that this result is sufficient to prove \Cref{mds_integer}. However, we view the entry distribution in \Cref{main} (uniform in $\{-m,\ldots,m\})$ as more natural for coding applications. In fact, we conjecture that any entry distribution that doesn't give too much probability to any specific element would do, see \Cref{conj_general} for the details.

\subsection{Proof techniques}
We prove \Cref{main} following the approach of Rudelson and Vershynin \cite{rudelson2008littlewood,rudelson2010non}, in particular following the lecture notes of Rudelson \cite{rudelson2013lecture} modified appropriately to handle the case of large $m$. On the other hand, Vempala {et al.} \cite{vempala2020communication} follow the approach of Kahn {et al.} \cite{kahn1995probability} and Tao and Vu \cite{tao2007singularity}. We briefly discuss the difference between the two approaches below.

At a high level, both approaches aim to study ``approximate periodicity'' of random vectors. However, they take different routes. The first approach is more direct, using the notion of Lowest Common Denominator (LCD) to define and study approximate periodicity. The second approach is indirect, using Fourier analysis. Fourier analytic techniques seem useful when the underlying entry distribution has well-behaved Fourier tails; for example, this is the case for the distribution considered in \cite{vempala2020communication}, where the entries are a sum of $m$ independent copies of a distribution over $\{-1,0,1\}$. However, the distribution we consider in this paper, uniform in $\{-m,\ldots,m\}$, has less well-behaved Fourier tails, and Fourier analytic techniques seem less suited to analyze it.

\subsection{Directions for further research and applications}

\paragraph{Singularity of matrices over general distributions.}
As we discussed above, most works on the singularity of random matrices give a bound on the singularity of $c^n$ for some absolute constant $c \in (0,1)$. \Cref{main} shows that if the entries are uniformly sampled from $\{-m,\ldots,m\}$, we can take $c=1/\text{poly}(m)$. We speculate that this is an instance of a much more general phenomena - the singularity probability is determined by the anti-concentration of the underlying entries distribution. Given a distribution $\D$ over $\R$, define its max-probability as $\|\D\|_{\infty} = \max_x \D(x)$. For example, if $\D$ is the uniform distribution over $\{-m,\ldots,m\}$, then $\|\D\|_{\infty}=1/(2m+1)$.

\begin{conjecture}\label{conj_general}
Let $\D$ be a distribution over $\R$ and set $p=\|D\|_{\infty}$. Let $M$ be a random $n \times n$ matrix with independent entries from $\D$. Then for some absolute constant $c>0$,
$$
\Pr[M \text{ is singular}] \le p^{cn}.
$$
\end{conjecture}

One can even speculate a more general conjecture, where each entry comes from a different underlying distribution, as long as they all have bounded max-probability.

\paragraph{Applications to coding theory.}
We view \Cref{mds_integer} as a first step towards the study of codes over integers. There are many intriguing questions that arise in coding theory, once we showed that random integer matrices are MDS with high probability. 
\begin{itemize}
    \item \textbf{Explicit constructions.} A natural question is to give an explicit construction of MDS matrices over integers with small integer values. Concretely, when $k=\alpha n$ for some constant $\alpha \in (0,1)$, to give an explicit construction of a $k \times n$ MDS matrix with a constant alphabet size (namely, independent of $n$).

    \item\textbf{Algorithms.} The next natural question, once there are explicit constructions, is to design efficient decoding algorithms for such codes. In particular, it would be intriguing to see if the smaller alphabet size can be utilized to obtain improved runtime (even by logarithmic factors).
\end{itemize}

\paragraph{Paper Outline.}
We prove \Cref{main} in the remainder of the paper. We start with a high-level overview of our framework in \Cref{sec:overview}.
We compute some preliminary estimates in \Cref{sec:estimates}, define and study incompressible vectors in \Cref{sec:compressible}, define the LCD condition in \Cref{sec:LCD}, where we also prove some properties of it, and bound the LCD of random vectors in \Cref{sec:LCD_bound}. We put all the ingredients together and complete the proof in \Cref{sec:proof_complete}.

\section{General approach}
\label{sec:overview}

We will follow the general approach of Rudelson \cite{rudelson2013lecture} with several modifications needed to obtain effective bounds for large $m$.

\paragraph{Notation.}
It will be convenient to scale the entries to be in $[-1,1]$; we denote by $\D$ the uniform distribution over $\{a/m: a \in \{-m,\ldots,m\}\}$. We denote by $\D^n$ the distribution over $n$-dimensional vectors with independent entries from $\D$, and by $\D^{n \times n}$ the distribution over $n \times n$ matrices with independent entries from $\D$.
We denote by $S^{n-1}$ the unit sphere in $\R^n$, namely $S^{n-1} = \{x \in \R^n: \|x\|_2=1\}$.
We will use the $c,c',c_0$, etc, to denote unspecified positive constants. Note that the same letter (e.g. $c$) can mean different unspecified constants in different lemmas.

\paragraph{We may assume that $n,m$ are large enough.}
We will assume throughout the proof that $n,m$ are large enough; concretely, for any absolute constants $n_0, m_0$, we may assume that $n \ge n_0, m \ge m_0$, and this would only effect the value of the constant $c$ in \Cref{main}.

To see why, consider first the regime of constant $m$ and large $n$. The distribution $\D$ is symmetric and bounded in $[-1,1]$. The results of \cite{rudelson2008littlewood} show that in such a case,
$$
\Pr[M \text{ is singular}] \le c^n
$$
for some absolute constant $c \in (0,1)$. This proves \Cref{main} for any constant $m$.

The other regime is that of constant $n$ and large $m$. While we may appeal to the result of Katznelson \cite{katznelson1993singular} in this regime, which gives a sharp bound of $c_n m^{-n}$, there is a much simpler argument that gives a bound of the order of $1/m$ which is good enough to establish \Cref{main} in this regime.
As the determinant of an $n \times n$ matrix is a polynomial of degree $n$, the Schwartz-Zippel lemma \cite{schwartz1980fast,zippel1979probabilistic} gives 
$$
\Pr[M \text{ is singular}] = \Pr[\det(M)=0] \le \frac{n}{m}.
$$
In particular, for constant $n$ and large $m$, this probability is of the order of $1/m$, which is consistent with the claimed bound of \Cref{main} (taking $c < 1/n$).

\paragraph{General approach.}
Let $M \sim \D^{n \times n}$, and let $X_1,\ldots,X_n$ denote its rows. If $M$ is singular, then one of the rows belongs to the span of the other rows. By symmetry we have
$$
\Pr[M \text{ is singular}] \le n \cdot \Pr[X_n \in \text{Span}(X_1,\ldots,X_{n-1})].
$$

Let $X^*$ be any unit vector orthogonal to $X_1,\ldots, X_{n-1}$ (if there are multiple ones, choose one in some deterministic way). We call it a \textit{random normal vector}. We will shorthand $X=X_n$. Observe that $X,X^*$ are independent. Thus we can bound
$$
\Pr[M \text{ is singular}] \le n \cdot \Pr[\ip{X^*, X}=0].
$$
To do so, we will show that unless $X^*$ belongs to a set of ``bad'' vectors, then the above probability is at most $m^{-cn}$, and that the probability that $X^*$ is bad is also at most $m^{-cn}$.

\section{Preliminary estimates}
\label{sec:estimates}

We establish some preliminary estimates in this section, which will be needed later in the proof.

\paragraph{Maximal eigenvalues of random matrices.}

The first ingredient is bounding the spectral norm of $M$. In fact, we would need this bound for rectangular matrices.
Given an $n \times k$ matrix $R$ we denote its spectral norm as $\|R\| = \max \{\|R x\|_2: x \in S^{k-1}\}$. Note that $\|R\|=\|R^T\|$ since
$\|R\| = \max_{x \in S^{k-1},y \in S^{n-1}} y^T R x$.

The following claim is a special case of \cite[proposition 4.4]{rudelson2013lecture}, who showed that it holds for any symmetric distribution $\D$ supported in $[-1,1]$.

\begin{claim}
\label{largest}
Let $R \sim \D^{n \times k}$ for $n \ge k$. Then for any $\lambda>0$,
$$
\Pr[\|R\| \ge \lambda \sqrt{n}] \le 2^{-c \lambda^2 k}.
$$
\end{claim}

\paragraph{Anti-concentration of projections.}

Next, we need anti-concentration results for projections of $\D^n$. To begin with we consider projections of the uniform distribution over the solid cube $[-1,1]^n$.

\begin{claim}\label{anticonc-uniform}
Let $U \sim [-1,1]^n$ be uniformly distributed. Then for every $x \in S^{n-1}$ and $\eps>0$,
$$
\Pr_u\left[|\ip{U,x}| \le \eps \right] \le  c \eps.
$$
\end{claim}

\begin{proof}
The uniform distribution $U \sim [-1,1]^n$ is a log-concave distribution. Let $S=\ip{U,x}$ and note that $S$ is a projection of $U$ along the direction $x$. The Prékopa–Leindler inequality \cite{leindler1972certain, prekopa1973logarithmic} states that projections of log-concave distributions are log-concave, and so $S$ is a log-concave distribution. Carbery and Wright \cite[Theorem 8]{carbery2001distributional} show that the required anti-concentration bound holds for any log-concave distribution.
\end{proof}

We extend this anti-concentration to the discrete case using a coupling argument. Here and throughout, we denote by $\log(\cdot)$ logarithm in base $2$.

\begin{claim}\label{anticonc}
Let $X \sim \D^n$ and set $\eps_0 = \frac{\sqrt{\log m}}{m}$. Then for every $x \in S^{n-1}$ and $\eps \ge \eps_0$, 
$$
\Pr\left[|\ip{X,x}| \le \eps\right] \le c \eps.
$$
\end{claim}

\begin{proof}
We apply a coupling argument between the uniform distribution in $[-1,1]^n$ and $\D^n$. Sample $X \sim \D^n$, $Y \sim [-1,1]^n$ and set
$Z = X + Y / 2m$. Observe that $Z$ is uniform in the solid cube $[-1-1/2m, 1+1/2m]^n$.
Next, fix $\eps>0$ and observe that $\ip{X,x} = \ip{Z,x} - \ip{Y,x} / 2m$. Thus we can bound
$$
\Pr[|\ip{X,x}| \le \eps] \le \Pr[|\ip{Z,x}| \le 2 \eps] + \Pr[|\ip{Y,x}| \ge 2 \eps m].
$$
For the first term, \Cref{anticonc-uniform} bounds its probability by $c_1 \eps$. For the second term, the Chernoff bound bounds its probability for $\eps \ge \eps_0$ by $1/m$. As we have $1/m \le \eps$, the claim follows.
\end{proof}

\paragraph{Tensorization lemma.}

We also need the following ``tensorization lemma'' \cite[Lemma 6.5]{rudelson2013lecture}.

\begin{claim}
\label{tensorization}
Let $Y_1,\ldots,Y_n$ be independent real-valued random variables. Assume for some $K, \eps_0>0$ that
$$
\Pr[|Y_i| \le \eps] \le K \eps \qquad \text{ for all } \eps \ge \eps_0.
$$
Then
$$
\Pr\left[ \sum_{i=1}^n Y_i^2 \le \eps^2 n\right] \le (c K \eps)^n \qquad \text{ for all } \eps \ge \eps_0.
$$
\end{claim}

\paragraph{Nets.}
A set of unit vectors $\N \subset S^{n-1}$ is called an $\eps$-net, for $\eps>0$, if it satisfies:
$$
\forall x \in S^{n-1} \; \exists y \in \N \; \|x-y\|_2 \le \eps.
$$ 
The following claim bounds the size of such a net. For a proof see for example \cite[Lemma 2.6]{milman2009asymptotic}.

\begin{claim}\label{net}
For any $\eps>0$, there exists a $\eps$-net $\N \subset S^{n-1}$ of size $|\N| \le (3/\eps)^n$.
\end{claim}

\paragraph{Integer points in ball.}

We need a bound on the number of integer vectors in a ball of a given radius. Let $B_n(r) = \{x \in \R^n: \|x\|_2 \le r\}$ denote the ball of radius $r$ in $\R^n$. The following bound is well known.

\begin{claim}\label{int_ball}
The number of integer vectors in $B_n(r)$ is at most $\left(1 + \frac{c r}{\sqrt{n}}\right)^n$.
\end{claim}

\section{Compressible vectors}
\label{sec:compressible}

The first set of ``bad'' vectors that we want to rule out are vectors which are close to sparse. A vector $u \in \R^n$ is $k$-sparse if it has at most $k$ nonzero coordinates.

\begin{definition}[Compressible vectors]\label{compress_defn}
Let $\alpha,\beta \in (0,1)$. A unit vector $x \in S^{n-1}$ is called $(\alpha,\beta)$-\emph{compressible} if it can be expressed as $x=u+v$, where $u$ is $(\alpha n)$-sparse and $\|v\|_2 \le \beta$. Otherwise, we say that $x$ is $(\alpha,\beta)$-\emph{incompressible}.
\end{definition}

We will later choose $\alpha,\beta$, but we note here that $\alpha$ will be a small enough absolute constant and $\beta$ a small polynomial
in $1/m$. Concrete values that work are $\alpha=1/50, \beta = 1/\sqrt{m}$. We will implicitly assume that both $n,m$ are large enough; concretely, at various places we assume that $\alpha n \ge 2$.

The main lemma we prove in this section is the following.

\begin{lemma}\label{compress}
Let $\alpha \in (0,1/8),\beta \in (\eps_0,1/2)$ where $\eps_0 = \frac{\sqrt{\log m}}{m}$. Then
$$
\Pr\left[X^* \text{ is } (\alpha,\beta)\text{-compressible}\right] \le (c \beta)^{n/8}.
$$
\end{lemma}

We need a bound on the smallest singular value of a rectangular matrix. 

\begin{claim}\label{vectoranti}
Let $R \sim \D^{n \times k}$ for $n \ge k$. Then for every $x \in S^{k-1}$ and $\eps \ge \eps_0$, 
$$
\Pr\left[\|Rx\|_2 \le \eps \sqrt{n}\right] \le (c \eps)^{n/2}.
$$
\end{claim}

\begin{proof}
Assume $\|Rx\|_2 < \eps \sqrt{n}$. This implies that $|(Rx)_i| \le 2\eps$ for at least $n/2$ coordinates $i \in [n]$.
Note that for each fixed $i$, the value $(Rx)_i$ is distributed as $\ip{X,x}$ for some $X \sim \D^k$.
Applying \Cref{anticonc} and the union bound over the choice of the $n/2$ coordinates gives
$$
\Pr\left[\|Rx\|_2 \le \eps \sqrt{n} \right] \le 2^n (c_1 \eps)^{n/2} = (c \eps)^{n/2}.
$$
\end{proof}

\begin{claim}
\label{smallsing}
Let $R \sim \D^{n \times k}$ for $n \ge 8k$. Then for every $\eps \ge \eps_0$, 
$$
\Pr\left[\min_{x \in S^{k-1}}\|Rx\|_2 \le \eps \sqrt{n} \right] \le \left(c \eps\right)^{n/4}.
$$
\end{claim}

\begin{proof}
We may assume that $\eps \le 1$ by taking $c \ge 1$.
Let $\N$ be an $(\eps^2)$-net in $S^{k-1}$ of size $|\N| \le (3/\eps^2)^k$, as given by \Cref{net}. Let $E_1$ denote the event that there exists $y \in \N$ for which $\|Ry\|_2 \le 2 \eps \sqrt{n}$. Applying \Cref{vectoranti} and a union bound gives
$$
\Pr\left[E_1\right] \le (3/\eps^2)^k \cdot (c_1 \eps)^{n/2} \le (c_2 \eps)^{n/4},
$$
where we used the assumption $n \ge 8k$.
Let $E_2$ denote the event that $\|R\| \ge \lambda \sqrt{n}$ for $\lambda=\sqrt{\log(1/\eps)}$.
\Cref{largest} shows that $\Pr[E_2] \le (c_3 \eps)^{n}$. We next show that if $E_1,E_2$ don't hold then the condition of the claim
also doesn't hold, namely that
$\|Rx\|_2 > \eps \sqrt{n}$ for all $x \in S^{k-1}$.

Let $x \in S^{k-1}$ be arbitrary and let $y \in \N$ be such that $\|x-y\|_2 \le \eps^2$. Then
$$
\|Rx\|_2 \ge \|Ry\|_2 - \|R\| \cdot \|x-y\|_2 \ge (2\eps - \eps^2 \lambda)\sqrt{n}.
$$
It can be verified that for $\eps \le 1$ we have $\eps \lambda \le 1$, which implies that $\|Rx\|_2 \ge \eps \sqrt{n}$.
\end{proof}

We will now use these two claims to prove \Cref{compress}.

\begin{proof}[Proof of \Cref{compress}]
Let $M'$ be the $(n-1) \times n$ matrix with rows $X_1,\ldots, X_{n-1}$. Assume that there exists an $(\alpha,\beta)$-compressible vector $x \in S^{n-1}$ in the kernel of $M'$. By definition, $x=u+v$ where $u$ is $(\alpha n)$-sparse and $\|v\|_2 \le \beta$. In particular, $M'(u+v)=0$ and hence $\|M' u\|_2 = \|M' v\|_2$.
In addition, $\|u\|_2 \ge \|x\|_2- \|v\|_2 \ge 1/2$ since $x$ is a unit vector and $\|v\|_2 \le \beta \le 1/2$.

Let $E$ denote the event that $\|M'\| \ge \lambda \sqrt{n}$ for $\lambda = c_1 \sqrt{\log(1/\beta)}$, where we choose $c_1 \ge 1$
large enough so that by \Cref{largest}, $\Pr[E] \le \beta^n$. Note that as we assume $\beta \le 1/2$ we have $\lambda \ge c_1 \ge 1$.
Assuming that $E$ doesn't hold, we have
$$
\|M' u\|_2 = \|M'v\|_2 \le \|M'\| \cdot \|v\|_2 \le \lambda \beta \sqrt{n}.
$$
In particular, $y=u / \|u\|_2$ is an $(\alpha n)$-sparse unit vector that satisfies $\|M' y\|_2 \le 2 \lambda \beta \sqrt{n}$. We next bound the probability that such a vector exists.

Let $\eps=2 \lambda \beta$, and note that $\eps \ge \eps_0$ since $\beta \ge \eps_0$ and $\lambda \ge 1$.
There are ${n \choose \alpha n}$ options for the support of $y$. Let $I=\{i: y_i \ne 0\}$ denote a possible support, set $k=|I|$ and let $R$ be an $(n-1) \times k$ matrix with columns $(Y_i: i \in I)$. As $\alpha < 1/8$ we have $n-1 \ge 8k$.
Thus we can apply \Cref{smallsing} and obtain that
$$
\Pr\left[\neg E \quad \wedge \quad \exists y \in S^{k-1}, \; \|Ry\|_2 \le \eps \sqrt{n}\right]
\le \left(c_2 \eps\right)^{n/4} = \left(c_3 \beta \sqrt{\log 1/\beta}\right)^{n/4}.
$$
Note that for $\beta \le 1$ we have $\beta \log(1/\beta) \le 1$ and hence the above bound is at most $(c_4 \beta)^{n/8}$.

To conclude, we union bound over the choices for $I$, the number of which is ${n \choose \alpha n} \le 2^n$. Thus we can bound
the total probability by $2^n (c_4 \beta)^{n/8} = (c_5 \beta)^{n/8}$.
\end{proof}

\section{The LCD condition}
\label{sec:LCD}
In this section we will introduce the notion of the Lowest Common Denominator (LCD) of a vector, which is a variant of the LCD definition in \cite{rudelson2013lecture}.
Informally, the LCD of a vector is a robust notion of ``almost periodicity''; concretely, it is the least multiple where most of its entries are close to integers.

Given $x \in \R^n$ let $x=[x]+\{x\}$ be its decomposition into integer and fractional parts, where $[x] \in \Z^n$ and $\{x\} \in [-1/2,1/2]^n$. 

\begin{definition}[Least common denominator (LCD)]
Let $\alpha,\beta \in (0,1)$. Given a unit vector $x \in S^{n-1}$, its least common denominator (LCD), denoted $\LCD_{\alpha,\beta}(x)$, is the infimum of $D >0$ such that we can decompose $\left\{Dx\right\}=u+v$, where $u$ is $(\alpha n)$-sparse and $\|v\|_{2} \le \beta \min(D,\sqrt{n})$.
\end{definition}

\begin{claim}
\label{compressible_LCD}
Assume $x \in S^{n-1}$ is $(5\alpha,\beta)$-incompressible. Then $\LCD_{\alpha,\beta}(x) > \sqrt{\alpha n}$.
\end{claim}

\begin{proof}
Let $D=\LCD_{\alpha,\beta}(x)$ and assume towards a contradiction that $D \le \sqrt{\alpha n}$. Let $y=Dx$. As $\|y\|_2^2 \le \alpha n$ there are at most $4 \alpha n$ coordinates $i \in [n]$ where $|y_i| \ge 1/2$. In all other coordinates $\{y_i\}=y_i$, and hence $y-\{y\}$ is $(4 \alpha n)$-sparse.
By assumption we can decompose $\{y\}=u+v$ where $u$ is $(\alpha n)$-sparse and $\|v\|_2 \le \beta D$. This implies that we can decompose $y=u'+v$ where $u'$ is $(5 \alpha n)$-sparse. Thus, we can decompose $x=y/D$ as $x=u''+v''$, where $u''=u/D$ is  $(5 \alpha n)$-sparse and $v''=v/D$ satisfies $\|v''\|_2 \le \beta$. This violates the assumption that $x$ is $(5\alpha,\beta)$-incompressible.
\end{proof}

Our main goal in this section is to prove the following lemma, which extends \Cref{anticonc} assuming $x$ has large LCD.
To get intuition, we note that the lemma below is useful as long as $\beta \ll \gamma \ll 1$. We will later set $\gamma = \sqrt{\beta}$ to be such a choice. In particular, if we set $\beta=m^{-1/2}$ then we have $\gamma=m^{-1/4}$.

\begin{lemma}
\label{anticonc-LCD}
Let $X \sim \D^n$.
Let $\alpha, \beta, \gamma \in (0,1/2)$, $x \in S^{n-1}$ be $(\alpha, \gamma)$-incompressible and set $D=\LCD_{\alpha,\beta}(x)$.
Then for every $\eps \ge 1/2\pi m D$, it holds that
$$
\Pr\left[|\ip{X,x}| \le \eps\right] \le
c \left( \frac{\eps}{\gamma} + \frac{1}{(\alpha \beta m)^{\alpha n}} \right).
$$
\end{lemma}

The proof of \Cref{anticonc-LCD} relies on Esseen's lemma \cite{esseen1966kolmogorov}.

\begin{lemma}[Esseen's Lemma]\label{esseen}
Let $Y$ be a real-valued random variable. Let $\phi_Y(t) = \E[e^{it Y}]$ denote the characteristic function of $Y$. Then
for any $\eps>0$, it holds that
$$
\Pr[|Y|\le \eps] \le c \eps\int_{-1/\eps}^{1/\eps} |\phi_Y(t)|dt.
$$
\end{lemma}

Before proving \Cref{anticonc-LCD}, we need some auxiliary claims.
Fix some $x \in S^{n-1}$, let $X \sim \D^n$ and let $Y=\ip{X,x}$. In order to apply \Cref{esseen}, we need to compute
the characteristic function of $Y$.

\begin{claim}
\label{fourier}
Let $X \sim \D^n$, $x \in S^{n-1}$ and set $Y=\ip{X,x}$. For $t \in \R$ it holds that
$$
|\phi_Y(t)| =  \prod_{k=1}^n F\left(\frac{x_k t}{2\pi m}\right)
$$
where $F:\R \to \R$ is defined as follows:
$$
F(y) = \left|\frac{\sin\left((2m+1)\pi y\right)}{(2m+1)\sin(\pi y)}\right|.
$$
\end{claim}

\begin{proof}
We have $Y=\sum x_i \xi_i$ where $\xi_1,\ldots,\xi_n \sim \D$ are independent. Hence
$$
\phi_Y(t) = \prod_{k=1}^n \E[e^{i x_k \xi_k t}].
$$
Next we compute
$$
\E[e^{i x_k \xi_k t}] = \frac{1}{2m+1}\sum_{\ell = -m}^{m}e^{i x_k (\ell/m) t}
= \frac{1}{2m+1}\cdot\frac{\sin(\frac{2m+1}{2m}x_k t)}{\sin(\frac{1}{2m} x_k t)}.
$$
Hence
$$
\left|\E[e^{i t x_k \xi_k}]\right| = F\left(\frac{x_k t}{2\pi m}\right).
$$
\end{proof}

The next claim proves some basic properties of the function $F$.

\begin{claim}
\label{F_prop}
The function $F$ satisfies the following properties:
\begin{enumerate}
\item $F$ is symmetric: $F(y)=F(-y)$ for all $y \in \R$.
\item $F$ is invariant to shifts by integers: $F(y)=F(\{y\})$ for $y \in \R$.
\item $F$ is bounded: $F(y) \in [0,1]$ for all $y \in \R$.
\item $F(y) \le G(m y)$ for $y \in [0,1/2]$, where $G:\R_{+} \to [0,1]$ is defined as follows:
$$
G(y) =
\begin{cases}
e^{-\eta y^2} & \text{if } y \in [0,1]\\
\frac{e^{-\eta}}{y} & \text{if } y \ge 1
\end{cases}
$$
Here, $\eta>0$ is an absolute constant. Note that $G$ is decreasing.
\end{enumerate}
\end{claim}

\begin{proof}
The first three claims follow immediately from the definition of $F$ in \Cref{fourier}. In order to prove the last claim, we will prove that $F(y) \le \frac{c_1}{ my}$ for $y \in [1/m,1/2]$ for some $c_1 \in (0,1)$; and that $F(y) \le \exp(-c_2(my)^2)$ for $y \in [0,1/m]$ for some $c_2>0$. The claim then follows by
taking $\eta = \min(\ln(1/c_1),c_2)$.

First, note that $F(y) \le \frac{1}{(2m+1)|\sin(\pi y)|}$. Using Taylor expansion at $0$, we get for $y \in [0,1/2]$ that
$$
\sin\left(\pi y\right) \ge \pi y-\frac{\pi^3y^3}{6} \ge \frac{\pi y}{2}.
$$
In particular, $F(y) \le \frac{1}{\pi my}$, which gives the desired bound for $c_1=1/\pi$.

Next, note that $F(y) = \frac{1}{2m+1}|\sin((2m+1)\pi y)\cdot\csc(\pi y)|$. The Laurent series of $\csc(x)$ at $x \ne 0$ is $\csc(x) = \frac{1}{x} + \frac{x}{6} + \frac{7x^3}{360} + \frac{31x^5}{15120} + \Theta(x^7)$  and the Taylor series for $\sin(x)$ is $\sin(x) = x - \frac{x^3}{3!} + \frac{x^5}{5!} + \Theta(x^7)$. So for $y \in [0,1/m]$ we have $F(y) \le 1 - c_2(my)^2 \le \exp(-c_2 (my)^2)$.
\end{proof}

We also need the following claim, which shows that incompressible vectors retain a large fraction of their norm when restricted to small coordinates. We use the following notation: given $x \in \R^n$ and a set of coordinates $J \subset [n]$, we denote by $x|_J \in \R^J$ the restriction of $x$ to coordinates in $J$.

\begin{claim}\label{spread}
Let $x \in S^{n-1}$ be $(\alpha,\gamma)$-incompressible.
Let $J = \left\{i: x_i \le \frac{1}{\sqrt{\alpha n - 1}}\right\}$. Then
$$
\|x|_J\|_2^2 \ge \|x|_J\|_{\infty}^2 + \gamma^2.
$$
\end{claim}

\begin{proof}
Let $J^c = [n] \setminus J$.
Since $x$ is a unit vector, we have $|J^c| \le \alpha n - 1$. Let $j \in J$ be such that $|x_j|$ is maximal
and take $K = J \setminus \{j\}$. Then $|K^c| \le \alpha n$, and since we assume that $x$ is $(\alpha,\gamma)$-incompressible,
we  have $\|x|_K\|_2 \ge \gamma$. This completes the proof, since
$$
\|x|_J\|_2^2 -\|x|_J\|_{\infty}^2  = \|x|_J\|_2^2 -x_j^2 = \|x|_K\|_2^2 \ge \gamma^2.
$$
\end{proof}

We would need the following lemma in the computations later on.

\begin{lemma}\label{firstint}
Let $\gamma, \delta>0$.
Let $x \in \R^n$ be a vector such that $\|x\|_{\infty} \le \delta$ and $\|x\|_2^2 \ge \|x\|_{\infty}^2 + \gamma^2$.
Let $T = \pi m / \delta$.
Then
$$
I = \int_{0}^{T} \prod_{i=1}^n F\left(\frac{x_i t}{2\pi m}\right) dt \le \frac{c}{\gamma}.
$$
\end{lemma}

\begin{proof}
To simplify the proof, we may assume by \Cref{F_prop}(1) that $x_i \ge 0$ for all $i$. Reorder the coordinates of $x$
so that $x_1 \ge x_2 \ge \ldots \ge x_n \ge 0$. Observe that for $x_i \in [0,T]$ we have $\frac{x_i t}{2 \pi m} \in [0,1/2]$ and hence
we can apply \Cref{F_prop}(4) and bound each term by $F\left(\frac{x_it}{2\pi m}\right) \le G\left(\frac{x_i t}{2\pi}\right)$. Thus
$$
I \le \int_{0}^{T} \prod_{i=1}^n G\left(\frac{x_i t}{2\pi}\right) dt =
2 \pi \int_{0}^{T / 2\pi} \prod_{i=1}^n G(x_i t) dt
\le 2 \pi \int_{0}^{\infty} \prod_{i=1}^n G(x_i t) dt.
$$

We  bound this last integral.
Let $t_i = 1 /x_i$ so that $t_1 \le t_2 \le \ldots \le t_n$. For simplicity of notation set $t_0=0, t_{n+1}=\infty$.
We break the computation of the integral to intervals $[t_k, t_{k+1})$ for $k=0,\ldots,n$, and denote by $I_k$ the integral in each interval:
$$
I_k = \int_{t_k}^{t_{k+1}} \prod_{i=1}^n G(x_i t) dt
= \int_{t_k}^{t_{k+1}} \prod_{i=1}^k \frac{e^{-\eta}}{x_i t} \cdot \prod_{i=k+1}^n e^{-\eta t^2 x_i^2} dt
= e^{-\eta k} \int_{t_k}^{t_{k+1}} \frac{e^{-\eta t^2 \sum_{i=k+1}^n x_i^2}}{t^k \prod_{i=1}^k x_i} dt.
$$

Fix $k$ and consider first the case that $\sum_{i=k+1}^n x_i^2 \ge \gamma^2/2$. In this case, using the fact that $x_i t  \ge 1$ for $i \in [k]$ and $t \in [t_k, t_{k+1}]$, we can bound $I_k$ by
$$
I_k \le e^{-\eta k} \int_{t_k}^{t_{k+1}} e^{-\eta \gamma^2 t^2/2} dt \le e^{-\eta k} \int_{0}^{\infty} e^{-\eta \gamma^2 t^2/2} dt \le \frac{c_1 e^{-\eta k}}{\gamma}.
$$

Next, consider the case that $\sum_{i=k+1}^n x_i^2 < \gamma^2/2$, which
means that $\sum_{i=1}^k x_i^2 > \|x\|_2^2 - \gamma^2/2 \ge \|x\|_{\infty}^2 + \gamma^2/2$. Observe that this is impossible for $k=0$ or $k=1$,
and hence we may assume $k \ge 2$.
In this case we bound
$$
I_k \le e^{-\eta k} \int_{t_k}^{t_{k+1}} \frac{1}{t^k\prod_{i=1}^k x_i }  dt
\le e^{-\eta k} \int_{t_k}^{\infty} \frac{1}{t^k\prod_{i=1}^k x_i} dt
= \frac{e^{-\eta k} x_k^{k-1}}{(k-1) \prod_{i=1}^k x_i}
\le \frac{e^{-\eta k}}{(k-1) x_1}.
$$
By our assumption, $\sum_{i=1}^k x_i^2 \ge \gamma^2/2$ and hence $x_1^2 \ge \gamma^2/2k$. Thus we can bound
$$
I_k \le \frac{e^{-\eta k}}{(k-1)\gamma/ \sqrt{2k}} \le \frac{c_2 e^{-\eta k}}{\gamma}.
$$
Overall, we bounded the integral by
$$
I \le 2 \pi \sum_{k=0}^n I_k \le 2 \pi \max(c_1,c_2) \sum_{k=0}^n \frac{e^{-\eta k}}{\gamma} \le \frac{c}{\gamma},
$$
where we used the fact that $c_1,c_2,\eta>0$ are all absolute constants.
\end{proof}

Now we have all the ingredients to complete proof of \Cref{anticonc-LCD}.

\begin{proof}[Proof of \Cref{anticonc-LCD}]
Let $Y=\ip{X,x}$. \Cref{esseen} and \Cref{fourier} give the bound
$$
\Pr[|Y|\le \eps] \le c_1 \eps I,
$$
where $I$ is the following integral:
$$
I = \int_{0}^{1/\eps} \prod_{i=1}^n F\left(\frac{x_i t}{2\pi m}\right) dt.
$$
Let $T = \pi m\sqrt{\alpha n-1}$. We will bound the integral in the regime $[0,T]$ and $[T,1/\eps]$, and denote the corresponding integrals
by $I_1,I_2$.

Consider first the integral $I_1$ in $[0,T]$. Let $\delta = 1 / \sqrt{\alpha n-1}$
and take $J = \{i: x_i \le \delta\}$.
Observe that by \Cref{F_prop}(3), we can bound $F\left(\frac{x_i t}{2\pi m}\right) \le 1$ for $i \notin J$.
Thus
$$
I_1 = \int_{0}^{T} \prod_{i=1}^n F\left(\frac{x_i t}{2\pi m}\right) dt \le
\int_{0}^{T} \prod_{i \in J} F\left(\frac{x_i t}{2\pi m}\right) dt.
$$
Next, as we assume that $x$ is $(\alpha,\gamma)$-incompressible, \Cref{spread} gives that $\|x|_J\|_2^2 \ge \|x|_J\|_{\infty}^2 + \gamma^2$.
Applying \Cref{firstint} to $x|_J$, we bound the first integral by
$$
I_1 \le \frac{c_2}{\gamma}.
$$

Next, consider the second integral $I_2$ in $[T,1/\eps]$. We will apply the LCD assumption to uniformly bound the integrand in this range.
Given $t \in [T, 1/\eps]$, let $y(t) = \left \{ \frac{x t}{2 \pi m} \right\} \in [-1/2,1/2]^n$,
 $\beta(t) = \beta \min(t / \sqrt{n}, 1)$ and $J(t) = \{i \in [n]: |y(t)_i| \ge \beta(t)\}$. As $t \le 1/\eps \le 2 \pi m D$, we have that $\frac{t}{2 \pi m} \le D = \LCD_{\alpha,\beta}(x)$, and hence $|J(t)| \ge \alpha n$. Applying \Cref{F_prop}, we bound the integrand by
$$
\prod_{i=1}^n F\left(\frac{x_i t}{2\pi m}\right) =
\prod_{i=1}^n F(y_i) \le
\prod_{i \in J(t)} F(y_i) \le
\prod_{i \in J(t)} G(m y_i) \le
\prod_{i \in J(t)} G(m \beta(t)) \le
G(m \beta(t))^{\alpha n}.
$$
Following up on this, we have
$$
\beta(t) \ge \beta(T) = \beta \frac{\sqrt{\alpha n - 1}}{\sqrt{n}} \ge \beta\sqrt{\alpha/2}  \ge \alpha \beta,
$$
where we used the assumptions that $\alpha n \ge 2$ and $\alpha \le 1/2$. We may assume that $\alpha \beta m \ge 1$, otherwise the conclusion of the lemma is trivial. In that case we have by \Cref{F_prop}(4) that
$$
G(m \beta(t)) \le G(\alpha \beta m) \le \frac{1}{\alpha \beta m}.
$$
Thus we can bound the integral $I_2$ by
$$
I_2 = \int_{T}^{1/\eps} \prod_{i=1}^n F\left(\frac{x_i t}{2\pi m}\right) dt \le
\frac{1 / \eps}{(\alpha \beta m)^{\alpha n}}.
$$
Overall, we get
$$
\Pr[|Y|\le \eps] \le c_1 \eps I = c_1 \eps (I_1 + I_2) \le \frac{c_1 c_2 \eps}{\gamma} + \frac{c_1}{(\alpha \beta m)^{\alpha n}}.
$$
\end{proof}

\section{Bounding the LCD}
\label{sec:LCD_bound}

Our main goal in this section is to prove that a random normal vector $X^*$ has large LCD with high probability.
Let $M'$ denote the $(n-1) \times n$ matrix with rows $X_1,\ldots,X_{n-1}$.
Let $D_0 = \sqrt{\alpha n}$ and $D_1 = \beta (\alpha \beta m)^{\alpha n}$ in this section.

\begin{lemma}
\label{LCD_large}
Let $\alpha \in (0,1/40)$, $\beta \in (0,1/2)$ and 
$D \in (1, D_1)$. Then 
$$
\Pr[\LCD_{\alpha,\beta}(X^*) \le D] \le D^2 \left(1 / \alpha c\right)^n \beta^{cn} 
$$
for some absolute constant $c \in (0,1)$.
\end{lemma}

We set $\gamma = \sqrt{\beta}$ throughout the section. We first condition on a number of bad events not holding.
Define:
\begin{align*}
&E_1 = \left[ \|M\| \ge \sqrt{n \log (1/\beta)} \right]\\
&E_2 = \left[ X^* \text{ is } \left(5 \alpha, \beta\right)\text{-compressible} \right]\\
&E_3 = \left[ X^* \text{ is } \left(\alpha, \gamma\right)\text{-compressible} \right]
\end{align*}
Applying \Cref{largest} for $E_1$, and
\Cref{compress} for $E_2,E_3$, we get that
$$
\Pr[E_1 \text{ or } E_2 \text{ or } E_3] \le \beta^{cn}.
$$
Thus, we will assume in this section that none of $E_1,E_2,E_3$ hold.
Assuming $\neg E_2$,
\Cref{compressible_LCD} yields that $\LCD_{\alpha,\beta}(X^*) \ge D_0$. For $D \ge D_0$ define
$$
\SD = \left\{x \in S^{n-1}: \LCD_{\alpha,\beta}(x) \in [D, 2D] \text{ and } x \text{ is } \left(\alpha, \gamma\right)\text{-incompressible} \right\}.
$$

The following is an analog of Lemma 7.2 in \cite{rudelson2013lecture}.

\begin{claim}
\label{SDnet}
Let $D \ge D_0$ and set $\nu = 6 \beta \sqrt{n} / D$.
There exists a $\nu$-net $\ND \subset \SD$ of size
$$
|\ND| \le (D/\beta) \left(\frac{cD}{\sqrt{\alpha n}}\right)^n (1/\beta)^{\alpha n}.
$$
Namely, for each $x \in \SD$ there exists $y \in \ND$ that satisfies $\|x-y\|_2 \le \nu$.
\end{claim}

\begin{proof}
Let $x \in \SD$ and shorthand $D(x) = \LCD_{\alpha,\beta}(x)$. By definition, we can decompose $\{D(x) x\} = u+v$ where $u$ is $(\alpha n)$-sparse and $\|v\|_2 \le \beta \min(D, \sqrt{n}) \le \beta \sqrt{n}$.

Let $W$ denote the set of $(\alpha n)$-sparse vectors $w \in [-1/2,1/2]^n$ such that each $w_i$ is an integer multiple of $\beta$. Then $|W| \le \binom{n}{\alpha n} (1/\beta)^{\alpha n}$, and
there exists $w \in W$ such that $\|u - w\|_{\infty} \le \beta$, which implies $\|u-w\|_2 \le \beta \sqrt{n}$.
This implies that
$$
\|\{D(x) x\} - w\|_2 \le 2 \beta \sqrt{n}.
$$

Next, consider $[D(x) x] \in \Z^n$. As $|[a]| \le 2 |a|$ for all $a \in \Z$, we have $\|[D(x) x]\|_2 \le 2 D(x) \|x\|_2 \le 4D$. Let $Z = \{z \in \Z^n: \|z\|_2 \le 4D\}$. Then $[D(x) x] \in Z$, and \Cref{int_ball} bounds $|Z| \le \left(1 + \frac{c_1 D}{\sqrt{n}}\right)^n$. So there is $z \in Z$ such that
$$
\|D(x) x - z - w\|_2 \le 2 \beta \sqrt{n}.
$$

Next, let $R$ be set of integer multiples of $\beta$ in the range $[D, 2D]$, so that $|R| \le D / \beta$ and there exists $r \in R$ with $|D(x)-r| \le \beta$. As $\|x\|_2 = 1$ we have
$$
\|r x - z - w\|_2 \le 2 \beta \sqrt{n} + \beta \le 3 \beta \sqrt{n}.
$$

Finally, define the set
$$
Y = \{(z+w)/r: z \in Z, w \in W, r \in R\}.
$$
Then there exists $y \in Y$ such that
$$
\|x - y\|_2 \le 3 \beta \sqrt{n} / D = \nu/2.
$$
Take a maximal set $\ND \subset \SD$ which is $\nu$-separated. That is, for any $x',x'' \in \ND$ we have $\|x'-x''\|_2>\nu$. Note that
by maximality, $\ND$ is a $\nu$-net in $\SD$. Next, note that $|\ND| \le |Y|$, as any point $x \in \ND$ must be $(\nu/2)$-close to a distinct point in $Y$. To conclude, we need to bound $|Y|$. We have
$$
|Y| \le |W| |Z| |R| \le {n \choose \alpha n} (1/\beta)^{\alpha n} \cdot \left(1 + \frac{c D}{\sqrt{n}}\right)^n \cdot (D / \beta).
$$
As $D \ge D_0 = \sqrt{\alpha n}$ we can simplify $1 + \frac{c D}{\sqrt{n}} \le \frac{(c+1) D}{\sqrt{\alpha n}}$. We can trivially bound ${n \choose \alpha n} \le 2^n$. Hence
$$
|\ND| \le |Y| \le (D/\beta) \left(\frac{2(c+1) D}{\sqrt{\alpha n}}\right)^n (1/\beta)^{\alpha n}.
$$
\end{proof}

\begin{claim}
\label{normal_SD}
For any $D \in [D_0,D_1]$ we have
$$
\Pr\left[X^* \in \SD \text{ and } \neg E_1 \right] \le
D^2 \left(c / \alpha\right)^n \beta^{n/8}.
$$
\end{claim}

\begin{proof}
First, note that we may assume $\beta \le \beta_0$ for any absolute constant $\beta_0 \in (0,1)$, by choosing the constant $c>0$ large enough to compensate for that (namely, taking $c \ge 1/\beta_0$). In particular, setting $\beta_0 = 2^{-20}$ works.

If $X^* \in \SD$ then there exists $y \in \ND$ such that $\|X^* - y\|_2 \le \nu$ for $\nu = 6 \beta \sqrt{n} / D$.
By definition of $X^*$ we have $M' X^*=0$, and as we assume that $\neg E_1$ hold, we have
$$
\|M' y\|_2 \le \|M'\| \|X^* - y\|_2 \le \nu \sqrt{n \log (1/\beta)}.
$$
Set $\beta_1 = 6 \beta \sqrt{\log(1/\beta)}$. The assumption $\beta \le \beta_0$ implies that $\beta_1 \le \beta^{3/4}$.
Set $\delta = \beta^{3/4} \sqrt{n} / D$.
We will bound the probability that there exists $y \in \ND$ such that $\|M' y\|_2 \le \delta \sqrt{n}$.

Fix $y \in \ND$, let $X \sim \D^n$, and define $p(\eps)=\Pr[|\ip{X,y}|] \le \eps$.
As $y \in \ND \subset \SD$ we have that $y$ is $(\alpha,\gamma)$-incompressible, and hence we can apply \Cref{anticonc-LCD}, which gives
$$
p(\eps) \le c_1 \left(\frac{\eps}{\gamma} + \frac{1}{(\alpha \beta m)^{\alpha n}} \right) 
\qquad \text{for all } \eps \ge 1 / 2 \pi m D.
$$
Next, we restrict attention to only $\eps \ge \delta$, and note that in this regime the first term is dominant
(since $D \le D_1$ we have $\delta \ge \beta^{3/4} \sqrt{n}/D_1 \ge 1 / (\alpha \beta m)^{\alpha n}$).
We can then simplify the bound as
$$
p(\eps) \le \frac{c_2 \eps}{\gamma}  \qquad \text{for all } \eps \ge \delta.
$$
Applying \Cref{tensorization}, and recalling that we set $\gamma=\sqrt{\beta}$, gives 
$$
\Pr\left[\|M' y\|_2 \le \delta \sqrt{n}\right] \le \left( \frac{c_3 \delta}{\gamma} \right)^{n-1}
= \left( \frac{c_4 \beta^{1/4} \sqrt{n}}{D} \right)^{n-1}.
$$
Union bounding over all $y \in \ND$, using \Cref{SDnet} to bound its size, gives
\begin{align*}
\Pr[\exists y \in \ND, \|M' y\|_2 \le \delta \sqrt{n}] &
\le (D/\beta) \left(\frac{cD}{\sqrt{\alpha n}}\right)^n (1/\beta)^{\alpha n} \cdot
\left( \frac{c_4 \beta^{1/4} \sqrt{n}}{D} \right)^{n-1} \\
& \le D^2 \left(c_5 / \sqrt{\alpha}\right)^n \beta^{n/4 - \alpha n - 2}.
\end{align*}
Our assumption $\alpha < 1/40$ and the implicit assumption $\alpha n \ge 2$ imply that $\alpha n + 2 \le n/8$,
which simplifies the above bound to the claimed bound.
\end{proof}

We are now in place to prove \Cref{LCD_large}.

\begin{proof}[Proof of \Cref{LCD_large}]
We may assume that non of $E_1,E_2,E_3$ hold, as
the probability that any of them hold is at most $\beta^{c_1 n}$ for some absolute constant $c_1 \in (0,1)$. 
This in particular implies that $\LCD_{\alpha,\beta}(X^*) \ge D_0$.
Fix $D \in [D_0, D_1]$. As $D \le D_1$ we can applying \Cref{normal_SD} to $D_i = 2^i D_0$ as long as $D_i \le D/2$. 
Summing the results we get
$$
\Pr\left[\LCD_{\alpha,\beta}(X^*) \le D \text{ and } \neg E_1, \neg E_2, \neg E_3\right]
\le  (2D)^2 (c_2 / \alpha)^n \beta^{n/8}.
$$
Thus overall we have
$$
\Pr\left[\LCD_{\alpha,\beta}(X^*) \le D\right] 
\le \beta^{c_1 n} + (2D)^2 (c_2 / \alpha)^n \beta^{n/8}.
$$
The lemma follows by taking $c \in (0,1)$ small enough.
\end{proof}

\section{Completing the proof}
\label{sec:proof_complete}

We now prove \Cref{main}.

\begin{proof}[Proof of \Cref{main}]
Fix $\alpha = 1/50, \beta = 1/\sqrt{m}$ and assume $m \ge m_0$ for a large enough constant $m_0$ to be determined soon.
Let $D$ to be determined soon. \Cref{LCD_large} gives
$$
\Pr[\LCD_{\alpha,\beta}(X^*) \le D] \le D^2 (1 / \alpha c_1 )^n \beta^{c_1 n}.
$$
As $\alpha$ is constant, and using the choice $\beta=1/\sqrt{m}$, we can simplify the bound as follows. For a small enough constant $c \in (0,1)$, setting $D=m^{cn}$ and $c_2=1/\alpha c_1$, we have
$$
\Pr[\LCD_{\alpha,\beta}(X^*) \le m^{cn}] \le m^{2cn} c_2^n m^{-(c_1/2) n} \le c_2^n m^{-(c_1/2 - 2c) n} \le c_2^n m^{-cn}.
$$
Assuming $m \ge m_0$ for a large enough constant $m_0$, we can simplify this bound further as
$$
\Pr[\LCD_{\alpha,\beta}(X^*) \le m^{cn}] \le m^{-(c/2) n}.
$$

Next, assume $D = \LCD_{\alpha,\beta}(X^*) \ge m^{c n}$. In this case, \Cref{anticonc-LCD} for $\eps=1/2 \pi m D$ gives that
$$
\Pr[\ip{X^*, X} = 0 ] \le \Pr[|\ip{X^*, X}| \le \eps ] \le c_3 \left(\frac{\eps}{\gamma} + \frac{1}{(\alpha \beta m)^{\alpha n}} \right)
\le m^{-c' n}
$$
for some $c' \in (0,1)$. Overall we obtain the desired bound.
\end{proof}

\section*{Acknowledgement}
We would like to thank Roman Vershynin and Konstantin Tikhomirov for helpful discussions.  

\bibliographystyle{abbrv}
\bibliography{matrices}

\end{document}